\documentclass[conference]{IEEEtran}
\IEEEoverridecommandlockouts
\usepackage[T1]{fontenc}
\usepackage{cite}
\usepackage{amsthm}
\usepackage{bm}
\usepackage{amsmath}
\usepackage{float}
\usepackage{setspace}
\usepackage{amssymb}
\usepackage{stfloats}
\usepackage{cite}
\usepackage{ragged2e}
\usepackage[ruled,vlined,linesnumbered]{algorithm2e}
\usepackage{amsfonts}
\usepackage{mathrsfs}
\usepackage{amsmath,amsthm}
\usepackage{array,booktabs}
\usepackage{subfigure}
\usepackage{multirow}
\usepackage{cuted}
\usepackage{multicol}
\usepackage{graphicx}
\usepackage{subfigure}
\usepackage{graphicx,xcolor,bm}
\usepackage[bookmarks=false]{hyperref}
\usepackage{threeparttable}
\usepackage{dcolumn}
\usepackage{setspace}
\usepackage{makecell}
\usepackage{lipsum}
\usepackage{enumerate}
\usepackage{mathrsfs}
\usepackage{bbm}

\newtheorem{Prop}{Proposition}
\newtheorem{thm}{Theorem}

\usepackage{cite,bm}
\graphicspath{{figures/}}
\def\BibTeX{{\rm B\kern-.05em{\sc i\kern-.025em b}\kern-.08em
  T\kern-.1667em\lower.7ex\hbox{E}\kern-.125emX}}
\begin{document}
	\title{AEGIS: Risk-Budgeted Online Scheduling for Resilient Continuous Edge Inference}
	
	\author{
		\IEEEauthorblockN{
			Houyi Qi\IEEEauthorrefmark{1},
			Minghui Liwang\IEEEauthorrefmark{1},
			Sai Zou\IEEEauthorrefmark{2},
			Wei Ni\IEEEauthorrefmark{3},
			Wei Gong\IEEEauthorrefmark{1}
		}
		\IEEEauthorblockA
		{
			\IEEEauthorrefmark{1}~ Shanghai Research Institute for
			Intelligent Autonomous Systems, Tongji University, Shanghai, China \\
		}
		\IEEEauthorblockA
		{
			\IEEEauthorrefmark{2}~ College of Big Data and Information Engineering, Guizhou University, Guizhou, China \\
		}
		\IEEEauthorblockA
		{
			\IEEEauthorrefmark{3}~ School of Engineering, Edith Cowan University, Perth, Australia \\
		}
		
		Email: \{houyiqi, minghuiliwang, weigong\}@tongji.edu.cn, dr-zousai@foxmail.com, Wei.Ni@ieee.org\\
		Corresponding author: Minghui Liwang
	}
	
	\maketitle
	
\begin{abstract}
	Continuous edge inference requires sustained wireless and computing support across successive service instances. Under recurring channel degradation, transient edge overload, and multi-user contention, isolated deadline misses may accumulate into persistent service degradation. Existing schedulers mainly optimize instantaneous latency or per-timeslot utility and provide limited control over such cross-time effects. To address this issue, we propose \textit{AEGIS} (\textit{Adaptive Exposure-Governed Inference Scheduling}), a risk-budgeted online framework for service-level operational resilience. AEGIS regulates predicted deadline-risk exposure through dynamically replenished per-user risk budgets and establishes an explicit finite-horizon bound on cumulative admitted-risk exposure. One-step state estimation supports anticipatory delay and risk assessment, while the centralized bandwidth--computing allocation is transformed into an exact-potential formulation and solved through asynchronous coordinate updates. Simulation results under varying system scales and stress conditions show that adaptive replenishment improves service continuity and recovery relative to fixed replenishment, while risk budgeting reduces predicted-risk exposure compared with more aggressive scheduling alternatives.
\end{abstract}
	
	
	\section{Introduction}
Next-generation wireless networks are evolving from connectivity infrastructures into intelligent service platforms that continuously support applications such as video analytics, augmented-reality interaction, mobile sensing, and real-time intelligent assistance~\cite{el2025collaborative}. Their operation relies jointly on persistent wireless access and edge-computing availability, making service quality vulnerable to recurring channel degradation, workload surges, and multi-user resource contention. Consequently, resilience in future wireless systems should extend beyond maintaining connectivity toward sustaining critical service functionality under time-varying disruptions~\cite{mahmood2025resilient,reifert2023comeback}. This requirement is particularly important for continuous edge inference, where users repeatedly generate deadline-sensitive tasks and adverse communication--computing conditions may cause deadline misses to persist across successive service instances. Such temporal accumulation can transform isolated task failures into prolonged service degradation even without infrastructure outages. We therefore consider \emph{service-level operational resilience}, namely, the capability of online orchestration to preserve timely service under disruption, contain the persistence of deadline violations, and recover service capability as communication and computing conditions improve.

Existing studies on wireless edge inference largely follow a performance-oriented resource-optimization paradigm, including joint model splitting and resource allocation~\cite{Li2025Optimal}, batching and early exiting~\cite{liu2023resource}, timely communication--computation co-scheduling~\cite{shisher2025computation}, service-level objective (SLO)-aware scheduling~\cite{zhang2026enabling}, and learning-based orchestration~\cite{fan2025vehicular}. Although effective in improving latency, SLO satisfaction, or per-timeslot utility, these methods mainly optimize nominal or instantaneous performance and provide limited control over cross-time service degradation. Meanwhile, wireless-resilience studies have considered outage- and criticality-aware adaptation~\cite{reifert2023comeback}, but rarely address accumulated deadline-risk exposure in persistent edge-intelligence services. This leaves a gap between efficient edge inference and resilient orchestration, motivating a cross-time control state that records degradation exposure, recovers during benign periods, and regulates subsequent bandwidth--computing admission.

To address this gap, we propose \textit{AEGIS}, an \textit{Adaptive Exposure-Governed Inference Scheduling} framework for resilient continuous edge inference. AEGIS couples cross-time risk regulation with online communication--computing orchestration, enabling the scheduler to account for accumulated service degradation while adapting resource allocation to evolving network conditions. The main contributions are summarized as follows.

\noindent
$\bullet$ We formulate continuous edge inference from a service-level operational-resilience perspective, where recurring communication and computing stress may cause deadline failures to accumulate across tasks. We introduce a cross-time risk-budget model and derive a finite-horizon upper bound on the cumulative predicted-risk exposure of admitted tasks.

\noindent
$\bullet$ We develop an online bandwidth--computing scheduling framework that incorporates the evolving risk state into resource coordination. The resulting discrete scheduling problem admits an exact-potential structure, for which we establish equilibrium existence and finite-step convergence of the asynchronous improvement process.

\noindent
$\bullet$ We conduct extensive evaluations under varying system scales, resource conditions, and dynamic stress. The results demonstrate that AEGIS improves service continuity and limits persistent deadline degradation while maintaining a favorable balance between timely-service performance and risk exposure.
	
\section{System Model and Resilience Objective}

\begin{figure}[!t]
	\vspace{-0.0cm}
	\setlength{\abovecaptionskip}{-1 mm}
	\centering
	\includegraphics[width=1\columnwidth]{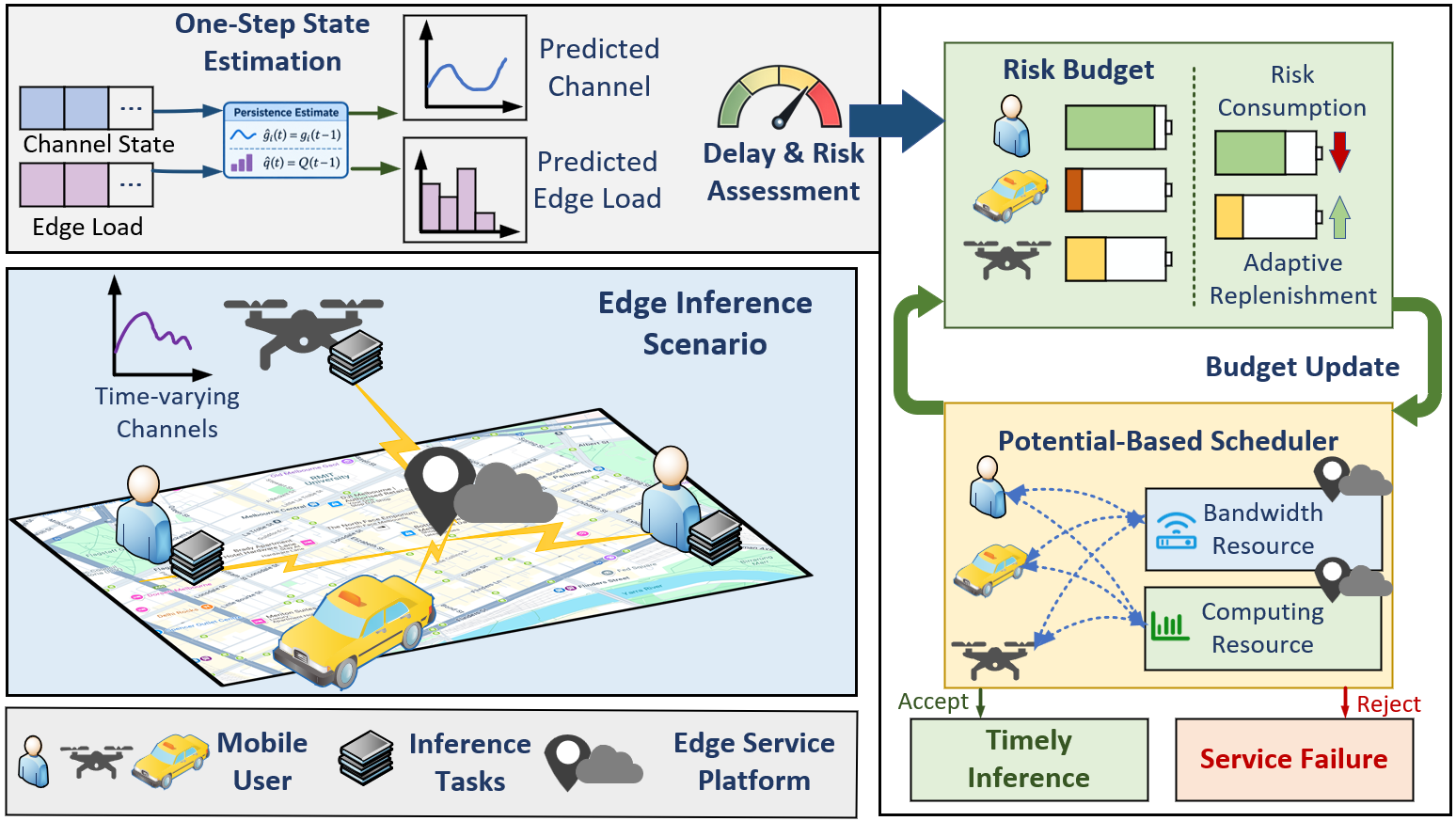}
	\caption{Architecture and closed-loop operation of AEGIS for resilient continuous edge inference, including one-step state estimation, predictive delay--risk assessment, adaptive risk-budget regulation, and potential-based communication--computing scheduling.}
	\vspace{-0.6cm}\label{fig:system_model}
\end{figure}

As shown in Fig.~\ref{fig:system_model}, we consider a wireless edge-inference system in which mobile users (MUs) continuously generate deadline-sensitive inference tasks and request service from an edge service platform (ESP). The ESP centrally determines task admission and jointly allocates uplink bandwidth and edge-computing resources. Under recurring channel degradation, transient edge overload, and multi-user contention, deadline failures may persist across successive tasks and develop into prolonged service degradation. AEGIS addresses this temporal effect through a closed-loop scheduling architecture: one-step estimates of channel and edge-load states support predictive delay and risk assessment, while replenishable user-level risk budgets regulate accumulated risk exposure across time. Based on the current risk states, a potential-based scheduler coordinates communication and computing resources and updates the budgets according to the resulting service decisions. Through this cross-time feedback, AEGIS aims to preserve timely-service capability, contain persistent service failures, and restore scheduling flexibility as adverse conditions subside.

\subsection{Continuous Edge Inference Task Model}

Let $\mathcal{U}=\{u_1,\dots,u_N\}$ denote the set of MUs. We consider a discrete time horizon $\mathcal{T}=\{1,\dots,T\}$, where each timeslot has an equal duration of $\tau$. To capture the temporal variation of inference demand, let $\chi_i(t)\in\{0,1\}$ denote the task-arrival indicator of MU $u_i$ in timeslot $t$, where $\chi_i(t)=1$ indicates the arrival of an active inference task and $\chi_i(t)=0$ otherwise. Accordingly, the set of active MUs is defined as $\mathcal{U}^{\mathrm{A}}(t)=\{u_i\in\mathcal{U}\mid\chi_i(t)=1\}$.
For each active MU $u_i\in\mathcal{U}^{\mathrm{A}}(t)$, its inference task is characterized by $\mathcal{J}_i(t)=\big(L_i(t),C_i(t),D_i^{\max}(t),w_i\big)$, where $L_i(t)$ denotes the input data size, $C_i(t)$ the required computing workload, $D_i^{\max}(t)$ the task deadline, and $w_i$ the service weight. The ESP determines the allocation action $a_i(t)=\big(b_i(t),f_i(t)\big)$, where $b_i(t)$ and $f_i(t)$ are the allocated uplink bandwidth and computing rate, respectively. Let $\tilde{\mathcal B}_i^{+}$ and $\tilde{\mathcal F}_i^{+}$ denote the positive candidate bandwidth and computing levels. The valid action set is
\begin{equation}
	\mathcal A_i(t)=\{(0,0)\}\cup\{(b,f)\mid b\in\tilde{\mathcal B}_i^{+},\ f\in\tilde{\mathcal F}_i^{+}\}.
	\label{eq:valid_action_set}
\end{equation}
The null action $(0,0)$ denotes task rejection, whereas each admitted action satisfies $b_i(t)>0$ and $f_i(t)>0$. The corresponding admission indicator is
\begin{equation}
	x_i(t)=\mathbbm{1}_{\left\{b_i(t)>0\ \wedge\ f_i(t)>0\right\}}.
\end{equation}

Each active MU generates at most one inference task at the beginning of each timeslot, with $D_i^{\max}(t)\le\tau$. An admitted task is counted as successfully served only if it completes within its deadline in the generating timeslot. A task that is rejected or misses its deadline is counted as a service failure and is discarded rather than carried over as residual workload. Hence, consecutive deadline violations represent failures of successive tasks generated by the same MU, which enables the system to characterize persistent service degradation across time.

\subsection{Delay and Timely Service Model}

For each admitted MU $u_i\in\mathcal{U}^{\mathrm{A}}(t)$, let $b_i(t)$ and $f_i(t)$ denote its allocated uplink bandwidth and computing rate, respectively. We assume orthogonal uplink bandwidth allocation among simultaneously admitted MUs. The wireless link varies over time according to
$g_i(t)=g_0\zeta_i(t)^{-\alpha_{\mathrm{pl}}}|h_i(t)|^2$,
where $\zeta_i(t)$ is the MU--ESP distance, $g_0$ is the reference-distance channel gain, and $\alpha_{\mathrm{pl}}$ is the path-loss exponent. To capture temporal channel correlation, the small-scale fading coefficient evolves as
$h_i(t)=\varrho h_i(t-1)+\sqrt{1-\varrho^2}\,e_i(t)$,
where $e_i(t)\sim\mathcal{CN}(0,1)$ and $\varrho\in[0,1)$ denotes the correlation coefficient.

Given transmit power $p_i$ and noise power spectral density $N_0$, the uplink signal-to-noise ratio (SNR) and transmission rate are
\begin{equation}
	\mathrm{SNR}_i(t)=\frac{p_i g_i(t)}{N_0 b_i(t)},~
	R_i(t)=b_i(t)\log_2\!\left(1+\mathrm{SNR}_i(t)\right).
	\label{eq:wireless_rate}
\end{equation}
The corresponding transmission delay is $d_i^{\mathrm{tx}}(t)=L_i(t)/R_i(t)$.
After transmission, the task may experience waiting and execution delays at the ESP. Let $Q(t)$ denote the background computing workload, measured in CPU cycles, that is already queued before the newly arriving inference tasks. Given the aggregate edge-computing capacity $F^{\mathrm{tot}}$, we approximate the resulting pre-service waiting time as
$d^{\mathrm{queue}}(t)=Q(t)/(F^{\mathrm{tot}}+\epsilon_f)$,
where $\epsilon_f>0$ provides numerical protection. Since $Q(t)$ represents workload already present before the current inference requests arrive, this waiting component is determined by the observed edge-load state rather than by the resource division among the newly admitted tasks. Once task of $u_i$ enters execution, its computing delay is $d_i^{\mathrm{comp}}(t)=C_i(t)/f_i(t)$.

Combining the communication and computing components, the end-to-end (E2E) delay of an admitted task is
\begin{equation}
	d_i(t)=d_i^{\mathrm{tx}}(t)+d^{\mathrm{queue}}(t)+d_i^{\mathrm{comp}}(t).
	\label{eq:e2e_delay_revised}
\end{equation}
Here, $Q(t)$ represents the background computing workload observed at the ESP before the current inference tasks enter service and determines their common pre-service waiting delay. An active task is successfully served if its E2E delay satisfies the deadline requirement, i.e., $\Psi_i(t)=\mathbbm{1}_{\{d_i(t)\le D_i^{\max}(t)\}}$. Inactive MUs with $\chi_i(t)=0$ do not participate in the current scheduling process and are assigned zero resources. For an active task that is not admitted, we set $\Psi_i(t)=0$, indicating a service failure in the current timeslot.
	
\subsection{Predictive Delay and Risk Assessment}

At the beginning of timeslot $t$, the ESP estimates the upcoming channel and edge-load states using observations available up to timeslot $t-1$. Considering the temporal correlation of both processes, we adopt one-step persistence estimates $\hat g_i(t)=g_i(t-1)$ and $\hat q(t)=Q(t-1)$. Based on these estimates, the ESP predicts the E2E delay associated with each candidate resource allocation.
Let $\bm a(t)=\big(a_1(t),\ldots,a_{|\mathcal U|}(t)\big)$ denote a candidate joint allocation profile. Based on the estimated channel and background-load states, the predicted E2E delay of active MU $u_i$ is
\begin{equation}
	\hat d_i(\bm a(t);t)
	=
	\hat d_i^{\mathrm{tx}}(\bm a(t);t)
	+
	\hat d^{\mathrm{queue}}(t)
	+
	\hat d_i^{\mathrm{comp}}(\bm a(t);t).
	\label{eq:predicted_delay}
\end{equation}
Specifically, the estimated channel state yields the predicted SNR
$\widehat{\mathrm{SNR}}_i(t)=p_i\hat g_i(t)/(N_0b_i(t))$,
and hence the predicted transmission delay is
$\hat d_i^{\mathrm{tx}}(\bm a(t);t)=L_i(t)/[b_i(t)\log_2(1+\widehat{\mathrm{SNR}}_i(t))]$.
The estimated background workload determines the predicted waiting delay
$\hat d^{\mathrm{queue}}(t)=\hat q(t)/(F^{\mathrm{tot}}+\epsilon_f)$,
while the allocated computing rate determines the predicted execution delay
$\hat d_i^{\mathrm{comp}}(\bm a(t);t)=C_i(t)/f_i(t)$.

To quantify the deadline feasibility of each candidate allocation, we first define the normalized deadline margin as
$\delta_i(\bm a(t);t)=[D_i^{\max}(t)-\hat d_i(\bm a(t);t)]/D_i^{\max}(t)$.
A positive margin indicates predicted completion before the deadline, whereas a smaller or negative value reflects increasing deadline pressure. We then transform this margin into a bounded deadline-risk score:
\begin{equation}
	\hat r_i(\bm a(t);t)
	=
	\frac{1}{1+\exp\!\big(\kappa_i\delta_i(\bm a(t);t)\big)},
	\label{eq:risk_sigmoid_new}
\end{equation}
where $\kappa_i>0$ controls the sensitivity of the risk score to the deadline margin. A larger $\hat r_i(\bm a(t);t)$ indicates greater predicted deadline pressure. For admitted actions, we define the predicted timeliness score as $\hat s_i(\bm a(t);t)=1-\hat r_i(\bm a(t);t)$ and the admitted-risk contribution as $r_i^{\mathrm{adm}}(\bm a(t);t)=\hat r_i(\bm a(t);t)$. Both quantities are set to zero for the null action $(0,0)$. The risk score is a normalized scheduling metric rather than a probabilistic estimate of deadline violation. 

To capture the recent operating risk beyond individual task evaluations, AEGIS constructs a system-level reference risk using a common resource configuration $a^{\mathrm{ref}}$. Let $\hat r_i^{\mathrm{ref}}(t)$ denote the predicted risk of active MU $u_i$ under this reference configuration. The mean reference risk is
\[
z(t)=\frac{1}{|\mathcal U^{\mathrm A}(t)|}
\sum_{u_i\in\mathcal U^{\mathrm A}(t)}
\hat r_i^{\mathrm{ref}}(t),
\]
which provides a comparable measure of the overall service condition across timeslots. If no MU is active, the latest available value of $z(t)$ is retained, with $z(0)=0$ for initialization. To reduce short-term fluctuations, we further compute the moving average
\begin{equation}
	\bar z_H(t)=\frac{\sum_{\ell=\max\{1,t-H+1\}}^{t}\hspace{-1mm}z(\ell)}{\min\{H,t\}}
	,
	~
	G_q(t)=\mathbbm{1}_{\left\{\bar z_H(t)<\theta_q\right\}},
	\label{eq:quantile_gate}
\end{equation}
where $H$ is the observation window and $\theta_q$ is a preset risk threshold. The indicator $G_q(t)$ identifies relatively benign operating periods and is subsequently used to regulate the recovery of the cross-time risk budget. 
	
\subsection{Cross-Time Risk Regulation and Long-Term Objective}

To regulate deadline-risk accumulation across successive tasks, AEGIS maintains a per-MU risk budget $B_i(t)\in[0,B_i^{\max}]$, where $B_i^{\max}$ denotes the maximum admissible budget. The budget represents the remaining tolerance for predicted deadline risk and constrains each admitted action by $\hat r_i(\bm a(t);t)\le B_i(t)$. We initialize $B_i(1)=B_i^{\max}$.

The budget is depleted by admitted risky service and replenished as the operating condition improves. Based on the system-level indicator $G_q(t)$ introduced above, the replenishment amount is
\begin{equation}
	\rho_i(t)=\rho_0+\kappa_{\mathrm{rec}}
	\frac{B_i^{\max}-B_i(t)}{B_i^{\max}}G_q(t),
	\label{eq:adaptive_replenishment}
\end{equation}
where $\rho_0\ge0$ provides regular replenishment and $\kappa_{\mathrm{rec}}\ge0$ controls the additional recovery when the recent system risk is relatively low. The resulting budget evolution is
\begin{equation}
	{\small\begin{aligned}
		B_i(t+1)=\min\Bigg\{B_i^{\max}, B_i(t)
		-\eta_i\chi_i(t)r_i^{\mathrm{adm}}(\bm a^\star(t);t)+\rho_i(t)\Bigg\},
	\end{aligned}}
	\label{eq:budget_update_new}
\end{equation}
where $\eta_i\in(0,1]$ controls the rate of risk consumption. Thus, higher-risk admissions deplete the available budget more rapidly and tighten subsequent admission, whereas lower-risk periods replenish the budget and recover scheduling flexibility. Since rejected tasks have $x_i(t)=0$, only admitted service contributes to risk-budget expenditure.

The budget dynamics in \eqref{eq:budget_update_new} also impose a finite-horizon bound on the cumulative predicted risk admitted by each MU.

\begin{Prop}[Finite-Horizon Admitted-Risk Accounting]
	\label{prop:risk_accounting}
	For any MU $u_i$, let
	$\omega_i(t)=\chi_i(t)r_i^{\mathrm{adm}}(\bm a^\star(t);t)$
	denote its admitted predicted-risk exposure in timeslot $t$. Under the budget dynamics in \eqref{eq:budget_update_new}, the cumulative exposure over any finite horizon $\{1,\ldots,T\}$ satisfies
	\begin{equation}
		{\small\begin{aligned}
				\eta_i\sum_{t=1}^{T}\omega_i(t)
				\le B_i(1)+\sum_{t=1}^{T}\rho_i(t)-B_i(T+1)\le B_i^{\max}+\sum_{t=1}^{T}\rho_i(t).
		\end{aligned}}
		\label{eq:finite_risk_bound}
	\end{equation}
	Accordingly, its time-average exposure satisfies
	$\frac{1}{T}\sum_{t=1}^{T}\omega_i(t)
	\le\frac{1}{\eta_iT}\big(B_i^{\max}+\sum_{t=1}^{T}\rho_i(t)\big)$.
\end{Prop}

The bound follows directly from the budget evolution in \eqref{eq:budget_update_new} by using
$\eta_i\omega_i(t)\le B_i(t)-B_i(t+1)+\rho_i(t)$
and summing over the considered horizon. Proposition~\ref{prop:risk_accounting} establishes that the accumulated predicted-risk exposure of admitted tasks is explicitly regulated by the available risk budget and its replenishment across time. This result characterizes the model-defined scheduling risk score and does not imply a probabilistic bound on realized deadline violations.

Based on the above risk-budget dynamics, we formulate the long-horizon scheduling problem to maximize the weighted timely-service performance over $\mathcal T$:
\begin{align}
	&\max_{\{\bm a(t)\}_{t\in\mathcal T}}\quad
	\sum_{t\in\mathcal T}\sum_{u_i\in\mathcal U^{\mathrm A}(t)}
	w_i\Psi_i(t)
	\label{eq:long_term_obj}\\
	\mathrm{s.t.}\quad
	&\sum_{u_i\in\mathcal U^{\mathrm A}(t)}\hspace{-2mm}b_i(t)\le B^{\mathrm{tot}},
	\hspace{-2mm}\sum_{u_i\in\mathcal U^{\mathrm A}(t)}\hspace{-2mm}f_i(t)\le F^{\mathrm{tot}}, ~\forall t\in\mathcal T
	\tag{C1}\label{con:resource}\\
	&r_i^{\mathrm{adm}}(\bm a(t);t)\le B_i(t),
	\quad \forall u_i\in\mathcal U^{\mathrm A}(t), ~\forall t\in\mathcal T
	\tag{C2}\label{con:risk}\\
	&a_i(t)\in\mathcal A_i(t),
	\quad \forall u_i\in\mathcal U^{\mathrm A}(t), ~\forall t\in\mathcal T
	\tag{C3}\label{con:action}\\
	&B_i(t+1)\ \text{satisfies \eqref{eq:budget_update_new}},
	\quad \forall u_i\in\mathcal U,~ \forall t\in\mathcal T
	\tag{C4}\label{con:budget}
\end{align}
where constraint (C1) enforces the shared bandwidth and computing capacities, while (C2) restricts each admitted task according to its current risk budget. Constraint (C3) specifies the feasible resource-allocation actions, and (C4) couples consecutive timeslots through risk consumption and replenishment. Inactive MUs are assigned the null action $(0,0)$ and do not participate in the current allocation.
The long-horizon formulation captures the cross-time impact of current scheduling decisions through the evolving risk budgets, resulting in a stochastic online scheduling problem with coupled communication, computing, and risk states. Conventional myopic or per-timeslot optimization cannot explicitly account for the future impact of current risk consumption, whereas direct full-horizon optimization requires unavailable future task arrivals, channel conditions, and edge workloads. We therefore transform the problem into a sequence of time-wise coordination decisions that jointly consider the current risk budgets and one-step state estimates, as developed in the next section.
	
\section{Design of AEGIS}

\subsection{Centralized Virtual-Agent Potential Formulation}
\label{subsec:MU_side_pg}

The ESP centrally determines admission and communication--computing resources for all active MUs. Since these decisions are coupled by shared capacities and per-MU risk budgets, directly searching the joint action space becomes costly as the number of active tasks increases. To obtain a structured online solution, the ESP associates each active MU $u_i$ with a virtual scheduling agent $v_i$ that controls only its resource-allocation coordinate. These virtual agents provide an algorithmic decomposition of the centralized problem rather than representing strategic behavior of physical MUs.
Let $\mathcal V(t)=\{v_i\mid u_i\in\mathcal U^{\mathrm A}(t)\}$ denote the active virtual-agent set. The resulting finite coordination game is $\mathcal G(t)=(\mathcal V(t),\{\mathcal A_i(t)\}_{v_i\in\mathcal V(t)},\{U_i(\cdot;t)\}_{v_i\in\mathcal V(t)})$, where $\mathcal A_i(t)$ is defined in \eqref{eq:valid_action_set}. We denote by $\mathcal A^{\mathrm{feas}}(t)$ the set of joint action profiles $\bm a(t)\in\prod_{v_i\in\mathcal V(t)}\mathcal A_i(t)$ that satisfy the shared bandwidth and computing constraints in (C1) and the risk-budget constraint in (C2) for every admitted task. Inactive MUs retain the null action $(0,0)$ and do not participate in the current coordination.

For each feasible joint action, the ESP evaluates the following time-wise coordination potential:
\begin{equation}
	{\small\begin{aligned}
		\Phi(\bm a(t);t)\hspace{-.8mm}=\hspace{-4.5mm}
		\sum_{u_i\in\mathcal U^{\mathrm A}(t)}\hspace{-1mm}
		\Bigg[\hspace{-.8mm}
		w_i\hat s_i(\bm a(t);t)
		\hspace{-.8mm}-\hspace{-.8mm}\alpha_i b_i(t)\hspace{-.8mm}-\hspace{-.8mm}\beta_i f_i(t)\hspace{-.8mm}-\hspace{-.8mm}\frac{\gamma_ir_i^{\mathrm{adm}}(\bm a(t);t)}{B_i(t)+\epsilon}\hspace{-.8mm}
		\Bigg],
	\end{aligned}}
	\label{eq:Phi_on_MU}
\end{equation}
where the first term rewards predicted timely service, $\alpha_i$ and $\beta_i$ penalize bandwidth and computing consumption, and $\gamma_i$ weights predicted-risk exposure relative to the remaining budget. Thus, $\Phi(\bm a(t);t)$ balances service timeliness, resource consumption, and cross-time risk pressure within each timeslot.
We define the payoff of virtual agent $v_i$ by its marginal contribution to the system potential:
\begin{equation}
	U_i(\bm a(t);t)
	=
	\Phi(\bm a(t);t)
	-
	\Phi\big((0,0),\bm a_{-i}(t);t\big).
	\label{eq:Ui_marginal_MU_on}
\end{equation}
Because a unilateral resource change must preserve the shared constraints, the feasible action set of $v_i$ given $\bm a_{-i}(t)$ is
\begin{equation}
	{\small\begin{aligned}
		\mathcal A_i^{\mathrm{feas}}(\bm a_{-i}(t);t)
		=
		\left\{
		a_i\in\mathcal A_i(t)\mid
		(a_i,\bm a_{-i}(t))\in\mathcal A^{\mathrm{feas}}(t)
		\right\}.
	\end{aligned}}
	\label{eq:Ai_feas_MU_on}
\end{equation}
A feasible joint action $\bm a^\star(t)$ is an equilibrium if no virtual agent can improve its payoff through a feasible unilateral change, i.e.,
$U_i(\bm a^\star(t);t)\ge U_i(a_i,\bm a_{-i}^\star(t);t)$ for all $a_i\in\mathcal A_i^{\mathrm{feas}}(\bm a_{-i}^\star(t);t)$. Such an equilibrium corresponds to a coordinate-wise local optimum of the centralized time-wise assignment, which provides the basis for the iterative solution developed next.
	
\subsection{Potential Property and Equilibrium Existence}

The marginal-contribution payoff in \eqref{eq:Ui_marginal_MU_on} aligns each feasible unilateral update with the centralized coordination potential, enabling an asynchronous improvement procedure for the time-wise scheduling problem.

\begin{thm}[Potential Property under Feasible Unilateral Deviations]
	\label{thm:MU_exact_potential}
	For any two feasible joint actions that differ only in the action of one virtual agent $v_i$, the corresponding payoff and potential variations satisfy
	\[
	U_i(\bm a'(t);t)-U_i(\bm a(t);t)
	=
	\Phi(\bm a'(t);t)-\Phi(\bm a(t);t).
	\]
	Hence, $\Phi$ is an exact potential over feasible unilateral deviations. Since $\mathcal A^{\mathrm{feas}}(t)$ is finite, at least one pure-strategy equilibrium under feasible unilateral deviations exists, and every strict feasible-improvement sequence terminates in finitely many updates~\cite{ding2021potential}.
\end{thm}

\begin{proof}
	The equality follows directly from the marginal-contribution payoff in \eqref{eq:Ui_marginal_MU_on}, since the reference term $\Phi((0,0),\bm a_{-i}(t);t)$ is unchanged under a unilateral deviation of $v_i$. Therefore, every strict feasible payoff improvement produces the same increase in $\Phi$. The remaining results follow from the finiteness of $\mathcal A^{\mathrm{feas}}(t)$.
\end{proof}

Theorem~\ref{thm:MU_exact_potential} guarantees finite termination at a feasible unilateral equilibrium when strict feasible-improvement updates continue until no improving action remains, providing the theoretical basis for the online scheduler.
	
\subsection{Online Scheduling Procedure of AEGIS}

\begin{algorithm}[t!]
	{\footnotesize \setstretch{0.54}
		\caption{AEGIS Online Scheduling}
		\label{Alg:AEGISOn}
		\LinesNumbered
		
		\textbf{Input:}
		$\mathcal U^{\mathrm A}(t)$, $\{\mathcal A_i(t)\}$, $B^{\mathrm{tot}}$, $F^{\mathrm{tot}}$, latest channel/load observations, current budgets $\{B_i(t)\}$, model parameters, $\theta_q$, $\varepsilon'$, and $K_{\max}$.
		
		Form $\hat g_i(t)$ and $\hat q(t)$, evaluate $G_q(t)$, and construct $\hat r_i(\bm a(t);t)$, $\hat s_i(\bm a(t);t)$, $\mathcal A^{\mathrm{feas}}(t)$, and $\Phi(\bm a(t);t)$.
		
		Initialize $k\leftarrow0$ and $\bm a^{(0)}(t)=\bm 0$.
		
		\While{$k<K_{\max}$}{
			For each $v_i\in\mathcal V(t)$, find
			$\tilde a_i^{(k)}(t)\in\arg\max_{a_i\in\mathcal A_i^{\mathrm{feas}}(\bm a_{-i}^{(k)}(t);t)}
			\Phi(a_i,\bm a_{-i}^{(k)}(t);t)$
			and compute
			$\Delta_i^{(k)}(t)=\Phi(\tilde a_i^{(k)}(t),\bm a_{-i}^{(k)}(t);t)-\Phi(\bm a^{(k)}(t);t)$.
			
			Set $\mathcal I^{(k)}(t)=\{i\mid\Delta_i^{(k)}(t)>\varepsilon'\}$.
			
			\If{$\mathcal I^{(k)}(t)=\varnothing$}{
				\textbf{break}
			}
			
			Select $i^\star\in\arg\max_{i\in\mathcal I^{(k)}(t)}\Delta_i^{(k)}(t)$.
			
			Update $a_{i^\star}^{(k+1)}(t)=\tilde a_{i^\star}^{(k)}(t)$ and retain all other actions.
			
			Set $k\leftarrow k+1$.
		}
		
		Set $\bm a^\star(t)=\bm a^{(k)}(t)$ and update $\{B_i(t+1)\}$ using \eqref{eq:budget_update_new}.
		
		\textbf{Output:}
		$\bm a^\star(t)$ and $\{B_i(t+1)\}_{u_i\in\mathcal U}$.
	}
\end{algorithm}

Algorithm~\ref{Alg:AEGISOn} summarizes the online implementation of AEGIS in each timeslot. The ESP first forms the one-step channel and background-load estimates, evaluates the system-level risk indicator, and constructs the predicted deadline-risk scores, feasible action space, and system potential. The joint allocation is initialized with the feasible null action $\bm a^{(0)}(t)=\bm 0$. At iteration $k$, each virtual agent $v_i$ keeps the allocations of the other MUs fixed and enumerates its feasible candidate actions in $\mathcal A_i^{\mathrm{feas}}(\bm a_{-i}^{(k)}(t);t)$. The ESP selects the candidate $\tilde a_i^{(k)}(t)$ that maximizes the resulting system potential and computes its improvement $\Delta_i^{(k)}(t)$. After evaluating all virtual agents, those satisfying $\Delta_i^{(k)}(t)>\varepsilon'$ constitute the improving set $\mathcal I^{(k)}(t)$. If $\mathcal I^{(k)}(t)$ is nonempty, the ESP accepts the update with the largest potential gain while keeping all other coordinates unchanged, and then proceeds to the next iteration.
The coordinate-update process continues until no feasible action provides an improvement larger than $\varepsilon'$ or the maximum number of iterations $K_{\max}$ is reached. The resulting joint action $\bm a^\star(t)$ specifies the admission, bandwidth, and computing decisions for the current timeslot, after which the corresponding risk budgets are updated according to \eqref{eq:budget_update_new}. Since every candidate action is evaluated within the current feasible set, all intermediate allocations satisfy the bandwidth, computing, and risk-budget constraints. Moreover, each accepted update strictly increases $\Phi$. Therefore, when $\varepsilon'=0$ and the procedure is not truncated by $K_{\max}$, Theorem~\ref{thm:MU_exact_potential} guarantees finite-step convergence to a pure-strategy equilibrium under feasible unilateral deviations.

Let $M_i=|\mathcal A_i(t)|$ denote the number of candidate actions of $v_i$. In each iteration, AEGIS examines at most $\sum_{v_i\in\mathcal V(t)}M_i$ feasible candidate actions. Hence, under the iteration limit $K_{\max}$, the online search requires at most $K_{\max}\sum_{v_i\in\mathcal V(t)}M_i$ candidate-potential evaluations. The one-step state estimation and system-level risk evaluation incur an additional $\mathcal O(|\mathcal V(t)|)$ operations per timeslot.
	
	\section{Evaluation}
We conduct simulations to evaluate the effectiveness of \textit{AEGIS}. All experiments are implemented in Python 3.10 on a 12th Gen Intel Core i9-12900H processor.

\subsection{Experimental Setup, Benchmarks, and Evaluation Metrics}

We consider a dynamic continuous edge-inference system with $N\in\{20,40,60,80,100\}$ MUs over $T=180$ one-second timeslots. Each episode consists of a nominal period (slots 1--60), a joint communication--computing stress period (slots 61--120), and a recovery period (slots 121--180). During stress, the channel power gains are attenuated by $6$ dB and the background computing workload is increased by a factor of $1.8$; the nominal processes resume afterward without resetting the risk budgets. We consider two resource regimes: \textit{fixed total resources} with $B^{\mathrm{tot}}=50$ MHz and $F^{\mathrm{tot}}=155$ GHz, and \textit{resources scaled with users} with $B^{\mathrm{tot}}(N)=50(N/40)$ MHz and $F^{\mathrm{tot}}(N)=155(N/40)$ GHz. In the latter case, the background workload is scaled by the same factor to preserve a comparable per-user operating condition. For each active MU, $L_i(t)\sim\mathcal{U}[0.12,0.90]$ MB, $C_i(t)\sim\mathcal{U}[0.08,0.95]\times10^9$ cycles, and $D_i^{\max}(t)\sim\mathcal{U}[0.28,0.82]$ s~\cite{qi2026Bank,3GPP2022,liu2023resource}. MU activity follows $\chi_i(t)\sim\mathrm{Bernoulli}(p_i)$, where $p_i\sim\mathrm{Beta}(5,4)$ is clipped to $[0.25,0.82]$. The cell radius is $250$ m with a minimum MU--ESP distance of $20$ m, and MU mobility follows a bounded random-waypoint model with speeds in $[1,8]$ m/s. We set the transmit power to $0.20$ W, noise PSD to $-174$ dBm/Hz, $g_0=10^{-4}$, $\alpha_{\mathrm{pl}}=3.5$, and $\varrho=0.95$. To capture temporal variations in edge congestion, the background workload is modeled as a two-state autoregressive process centered at $2.4\times10^{10}$ and $7.5\times10^{10}$ cycles for low- and high-load conditions, respectively, with autoregressive coefficient $0.88$, state-transition probabilities $0.06$ and $0.18$, and workload range $[5\times10^9,1.15\times10^{11}]$ cycles. Bandwidth and computing allocations are discretized as $\{1,2,4,6,8,10\}$ MHz and $\{2,4,6,8,10,12\}$ GHz, respectively, and we set $\alpha_i=0.040$, $\beta_i=0.015$, and $\gamma_i=0.25$. The service weight $w_i$ is sampled from $\{1.0,1.5,2.0\}$ with probabilities $\{0.5,0.3,0.2\}$, while the risk-sensitivity coefficient, improvement threshold, and iteration limit are set as $\kappa_i=6.0$, $\varepsilon'=10^{-12}$, and $K_{\max}=120$, respectively. For the quantile gate, we use $q=0.75$ with the corresponding threshold $\theta_q=0.5908$, which remains fixed across different system scales and resource regimes. All methods are evaluated over 30 matched seeds under identical task, mobility, channel, workload, and stress realizations.

We compare AEGIS with six representative schemes: \textit{(i) SLO-Edge}, an SLO-aware priority heuristic inspired by~\cite{zhang2026enabling}; \textit{(ii) ChannelFirst}, a channel-aware heuristic inspired by the instantaneous-efficiency principle in~\cite{liu2023resource}; \textit{(iii) EqualShare}, which continuously divides the available bandwidth and computing capacity among active MUs; \textit{(iv) Lyapunov-DPP}, a drift-plus-penalty controller using the same resource-action space and virtual risk debt $Z_i(t+1)=\max\{0,Z_i(t)+x_i(t)\hat r_i(t)-0.2\}$ with weight $1.6$; \textit{(v) AEGIS-FixedRefill}, which removes adaptive budget recovery by setting $\kappa_{\mathrm{rec}}=0$; and \textit{(vi) AEGIS-NoBudget}, which removes the cross-time risk budget while retaining the same state estimates and resource coordination procedure. These schemes cover priority-based scheduling, equal sharing, long-term risk control, and mechanism-level ablations of the proposed cross-time risk regulation.
We evaluate all methods using six complementary service, risk, and resilience metrics. \textit{Timely inference ratio (TIR)} is the fraction of active tasks completed within their deadlines, with rejected tasks counted as unsuccessful; \textit{average predicted deadline-risk score (APR)} measures the mean normalized risk of admitted tasks; \textit{admission ratio (AR)} is the fraction of active tasks admitted for edge execution; and \textit{deadline-violation burst length (DVBL)} measures the average persistence of consecutive service failures for the same MU.
To characterize recovery from transient stress, let $\bar y_m(t)$ denote the 10-slot moving-average TIR of method $m$. We use $\mathcal T^{\mathrm{pre}}=\{41,\ldots,60\}$ as the pre-stress reference interval, $\mathcal T^{\mathrm{rec}}=\{121,\ldots,180\}$ as the recovery interval, and $\mathcal T^{\mathrm{sr}}=\{61,\ldots,180\}$ as the stress--recovery interval. The corresponding pre-stress and recovery service levels are $y_m^{\mathrm{pre}}=|\mathcal T^{\mathrm{pre}}|^{-1}\sum_{t\in\mathcal T^{\mathrm{pre}}}\bar y_m(t)$ and $y_m^{\mathrm{rec}}=|\mathcal T^{\mathrm{rec}}|^{-1}\sum_{t\in\mathcal T^{\mathrm{rec}}}\bar y_m(t)$. We define the \textit{RecoveryRatio (RR)} as $\mathrm{RR}_m=y_m^{\mathrm{rec}}/(y_m^{\mathrm{pre}}+\epsilon_y)$ and the \textit{normalized service-degradation area (NSDA)} as $\mathrm{NSDA}_m=|\mathcal T^{\mathrm{sr}}|^{-1}\sum_{t\in\mathcal T^{\mathrm{sr}}}\max\{0,1-\bar y_m(t)/(y_m^{\mathrm{pre}}+\epsilon_y)\}$, where $\epsilon_y>0$ avoids numerical singularity. Higher TIR, AR, and RR indicate stronger service capability and recovery, whereas lower APR, DVBL, and NSDA indicate lower risk exposure and better degradation containment. All statistics and confidence intervals are computed at the matched-seed level.
	
\subsection{Performance Evaluation}
\begin{figure}[!t]
	\vspace{-0.0cm}
	\setlength{\abovecaptionskip}{-1 mm}
	\centering
	\includegraphics[width=0.485\textwidth]{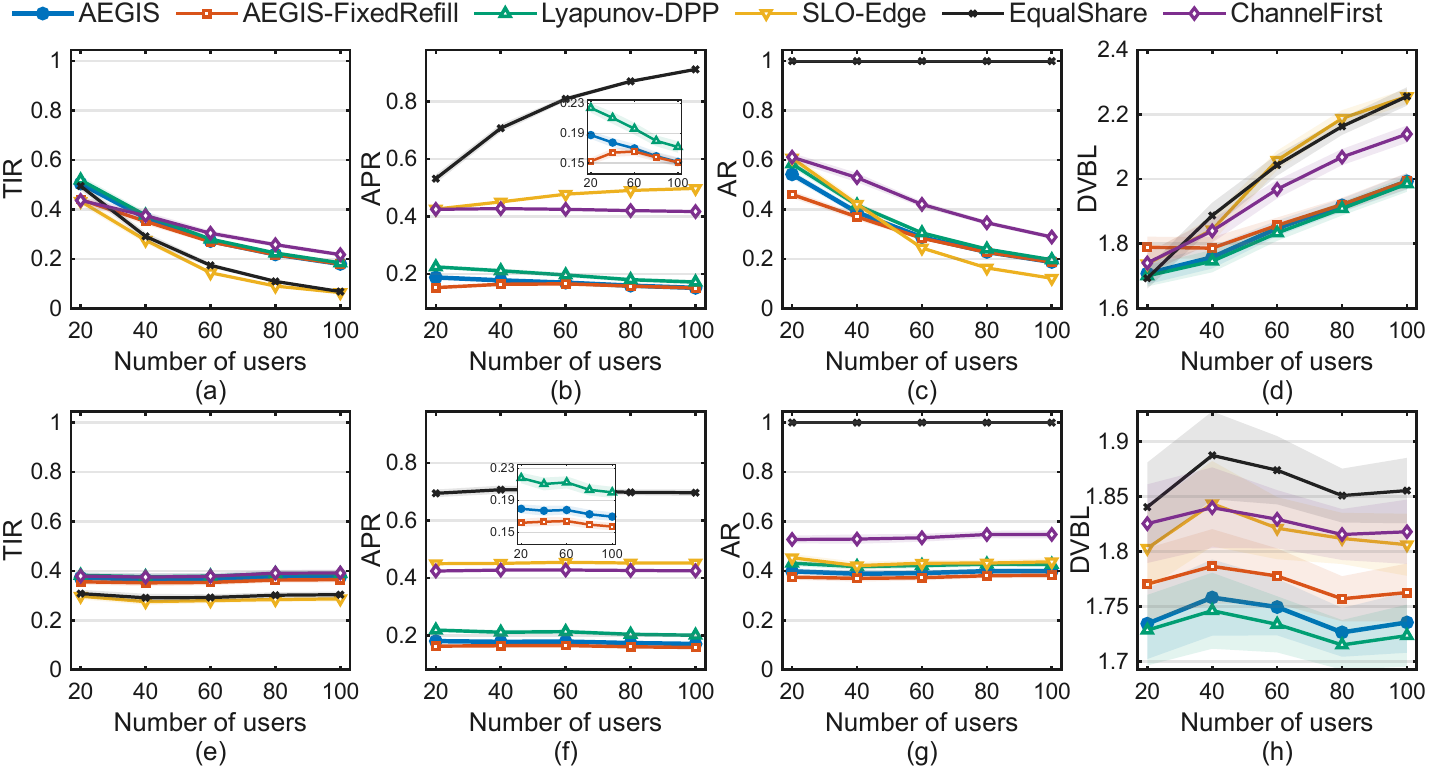}
\caption{Service--risk scaling with increasing MU population under fixed total resources in (a)--(d) and proportionally scaled resources in (e)--(h), with the four columns reporting TIR, APR, AR, and DVBL and shaded bands indicating 95\% confidence intervals.}
	\vspace{-0.4cm}
	\label{fig:scaling_results}
\end{figure}
We first examine the service--risk behavior under increasing system scale. Fig.~\ref{fig:scaling_results} reports TIR, APR, AR, and DVBL under two resource regimes. Panels (a)--(d) use fixed total resources, whereas panels (e)--(h) scale bandwidth and computing capacity proportionally with the MU population. All curves are averaged over 30 matched seeds with two-sided 95\% Student-$t$ confidence intervals.
Under fixed total resources, increasing the number of MUs intensifies communication--computing contention, leading to lower TIR and longer DVBL. AEGIS maintains higher TIR and markedly lower APR than SLO-Edge and EqualShare across the evaluated scales. ChannelFirst often achieves higher TIR by favoring strong channels, but incurs substantially larger APR. Lyapunov-DPP also provides slightly higher TIR, whereas AEGIS consistently operates at lower APR, reflecting a clear risk--service tradeoff. Compared with AEGIS-FixedRefill, AEGIS improves service continuity and reduces violation persistence at the cost of moderately higher APR.
When resources scale with the MU population, the four metrics remain comparatively stable as the number of MUs increases, indicating that proportional resource and workload scaling preserves an approximately scale-neutral operating condition. The relative behavior of the competing methods also remains consistent: AEGIS retains substantially lower APR than ChannelFirst and Lyapunov-DPP while maintaining competitive TIR and DVBL.

\begin{figure}[!t]
	\vspace{-0.0cm}
	\setlength{\abovecaptionskip}{-1 mm}
	\centering
	\includegraphics[width=\columnwidth]{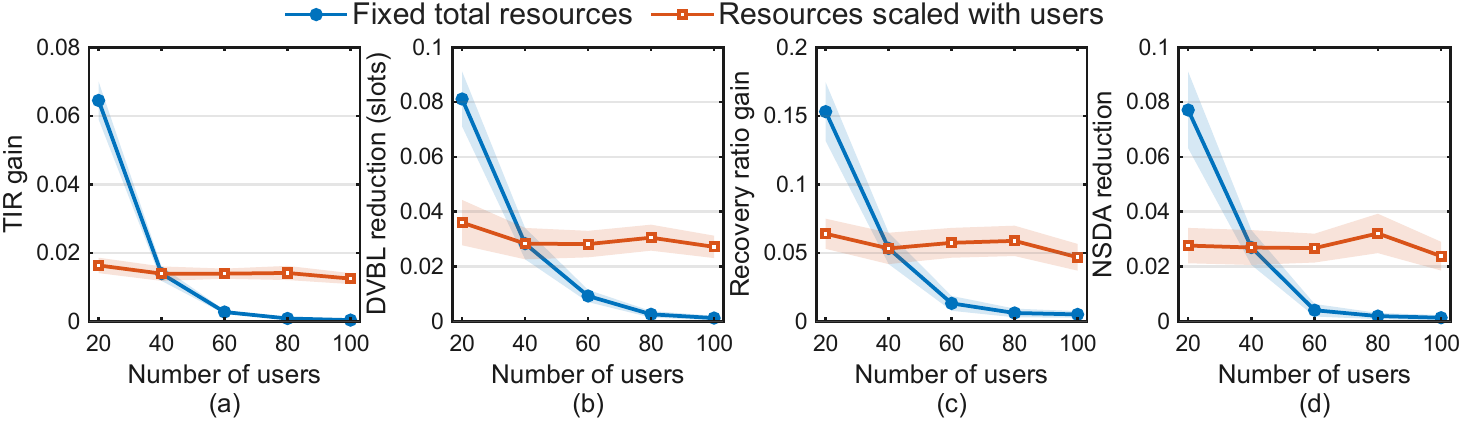}
\caption{Matched paired effects of AEGIS relative to AEGIS-FixedRefill for (a) TIR, (b) DVBL, (c) RR, and (d) NSDA, where positive values favor AEGIS and shaded bands denote two-sided 95\% Student-$t$ confidence intervals over 30 paired seeds.}
	\vspace{-0.4cm}
	\label{fig:dynamic_budget_effects}
\end{figure}

We next examine the effect of adaptive replenishment through matched comparisons with AEGIS-FixedRefill. Fig.~\ref{fig:dynamic_budget_effects} reports the paired TIR gain, DVBL reduction, RR gain, and NSDA reduction, with positive values favoring AEGIS.
As shown in Fig.~\ref{fig:dynamic_budget_effects}, adaptive replenishment consistently improves service continuity and degradation containment. Under fixed resources, the gain is most pronounced at moderate scale and diminishes as severe resource scarcity becomes dominant. Under proportionally scaled resources, the benefit remains visible across all evaluated population sizes, indicating that the mechanism is not tied to a specific congestion level.

\begin{table}[t]
	\caption{Paired comparison of AEGIS across ten scale--resource settings, where each entry reports the number of settings with a favorable paired 95\% confidence interval.}
	\label{tab:paired_mechanism_effects}
	\centering
	\scriptsize
	\setlength{\tabcolsep}{3.2pt}
	\renewcommand{\arraystretch}{1.08}
	\begin{tabular}{lccccc}
		\toprule
		Baseline & TIR$\uparrow$ & APR$\downarrow$ & DVBL$\downarrow$ & RR$\uparrow$ & NSDA$\downarrow$ \\
		\midrule
		AEGIS-FixedRefill & 10 & 0 & 10 & 10 & 10 \\
		AEGIS-NoBudget    & 0  & 10 & 0  & 0  & 0  \\
		Lyapunov-DPP      & 0  & 10 & 0  & 0  & 1  \\
		\bottomrule
	\end{tabular}\vspace{-0.4cm}
\end{table}

Table~\ref{tab:paired_mechanism_effects} summarizes the paired significance across the ten scale--resource settings. AEGIS consistently improves TIR, DVBL, RR, and NSDA over AEGIS-FixedRefill, whereas the fixed-refill variant remains more conservative in APR. Against AEGIS-NoBudget and Lyapunov-DPP, the most consistent advantage of AEGIS is lower APR, while the competing schemes retain slightly stronger service performance. These results further confirm the intended tradeoff between service continuity and cross-time risk control.

Overall, adaptive replenishment improves service continuity and recovery relative to fixed replenishment, while cross-time risk budgeting reduces predicted-risk exposure compared with more aggressive scheduling alternatives. These results highlight the tradeoff between timely-service performance and accumulated risk regulation.

\section{Conclusion}

We investigated service-level operational resilience for continuous edge inference under dynamic communication--computing conditions. We proposed AEGIS, which combines cross-time risk regulation with joint bandwidth--computing orchestration and establishes a finite-horizon bound on cumulative predicted-risk exposure. An exact-potential formulation further enables finite-step asynchronous resource coordination. Experiments show that AEGIS improves service continuity and degradation containment through adaptive risk-budget recovery, while maintaining lower predicted-risk exposure than more aggressive scheduling schemes. Overall, AEGIS provides a balanced approach to timely-service preservation and cross-time risk control for resilient continuous edge intelligence.

\section*{Acknowledgment}
This work was supported in part by Shanghai Pujiang Programme (24PJD117), Guizhou Provincial Science and Technology Plan Project ([2026]209), National Natural Science Foundation of China (62271424, 62403360).
	
	\begin{spacing}{0.9}
		\bibliographystyle{ieeetr}
		\bibliography{reference}
	\end{spacing}
	
\end{document}